\documentclass[a4paper,onecolumn,11pt,unpublished=2017-05-09]{quantumarticle}
\pdfoutput=1
\usepackage[utf8]{inputenc}
\usepackage[english]{babel}
\usepackage[T1]{fontenc}
\usepackage{amsmath}
\usepackage{hyperref}
\usepackage{graphicx}
\usepackage{amsthm}
\usepackage{amssymb}
\usepackage{amsfonts}
\usepackage[numbers,sort&compress]{natbib}

\usepackage{tikz}

\newcommand{\braket}[2]{\ensuremath{\left\langle#1\mid#2\right\rangle}}
\newcommand{\bhat}[1]{\boldsymbol{\textbf{#1}}}

\newtheorem{theorem}{Theorem}
\newtheorem{example}{Example}

\begin{document}

\title{A Quaternionic Map Causing Bipartite Entanglement}

\author[1]{Lidia Obojska}
\affiliation[1]{Department of Quantitative Methods and Information Technology, Kozminski University, Jagiellonska 57/59 St., 03-301 Warsaw, Poland;}
\affiliation[1]{  Department of Mathematics, Siedlce University of Natural Sciences and Humanities, 3 Maja 54 St., 08-110 Siedlce, Poland}
\orcid{0000-0001-6230-7305}
\email{lobojska@kozminski.edu.pl}

\maketitle
\begin{abstract}
 In the following manuscript we propose a quaternionic map which transforms a single quantum state into a bipartite entanglement. Until now such a transformation has not been defined yet. To define such a map, we embed one particle state within the algebra of complex quaternions. Next, on the basis of rotational features of quaternions, we choose a special quaternion which in combination with a quaternion describing a particle state, causes its splitting in a specified direction in 4D space. Finally, we prove that the proposed operation, under several restrictions, always causes bipartite entanglement.
\end{abstract}

\section{Introduction}\label{sec1}

Bipartite entanglement, i.e., the lack of separability,  is a fundamental concept in quantum information theory.  By inseparability we mean a quantum state that cannot be decomposed onto a tensor product of its local constituents; i.e., we do not have individual particles but an inseparable whole. This phenomenon results from an interaction between particles which Einstein defined: ``a spooky action at a distance''. Perhaps, it is the best known quantum phenomenon; however, until now its mathematical description does not exist \cite{EPR1935,Schrodinger1935,Bell1964,4Hor2007,Chenetall2008,Bell2011, Hall2013, Susskind2014}. In 1964 John S. Bell proved that one of the key assumptions in quantum mechanics (QM), the principle of locality, was mathematically inconsistent with the predictions of quantum theory \cite{Bell1964}. He demonstrated an upper limit, regarding the strength of correlations that can be produced in any theory obeying local realism, and showed that quantum theory predicts violations of this limit for certain entangled systems. The first experimental proof was due to Carl Kocher, who  presented an apparatus in which two photons successively emitted from a calcium atom were shown to be entangled; therefore, we know how to create entanglement by experiment \cite{Ko1967, LKS2015}, but we do not have its mathematical description. 

\bigskip In QM we apply the Dirac notation; however, the rotational properties of particles, are not so evident as in the case when quaternion description is used. The following work is an attempt to define a transformation in terms of complex quaternions, which causes bipartite entanglement. We limit ourselves only to bipartite entanglement, since the algebra of octonions is not associative, what causes some inconveniences for descriptions of composite quantum states. However, a mathematical description of two entangled particles can shed new light for more complex systems.

\bigskip
The presented approach is based on the fact that a subatomic particle decays into an entangled pair of other particles considered as indiscernible units; therefore we do not take two particles and correlate them in a specific way, but we take one quantum particle and split it creating an entangled state. This fact causes problems in classical mathematics represented by the set of ZFC axioms since indiscernible objects do not exist within this framework. In \cite{LO2019} we tried to fill this gap in the foundations of mathematics and proposed a modified mereology as a mathematical background which can serve in foundations of quantum mechanics. This theory is based on a relation of division, which is preordering. We showed how to define an order relation, and how to shift to classical ZFC set theory by the use of a special unitary operator, which transforms collective sets into distributive ones. This theoretical background is necessary if we want to define a splitting transformation in term of  complex quaternions.

\bigskip 
Complex quaternions are also known as biquaternions. They are widely applied in quantum mechanics in terms of the Pauli matrices \cite{SA1986,Dixon2013, Ward1997}. They are often used in quantum physics, but their rotational properties are rarely taken into consideration. It is known that biquaternions represent spinors, i.e., vectors with specified rotational properties; hence, their use for description of bipartite entanglement seems perfect \cite{Alexeyeva2007, Alexeyeva2012,Rodrigues1997}. In Section \ref{Sec:Quat} we present a brief description of real and complex quaternions. In particular, we focus on their rotational features. We show how to decompose rotations in two orthogonal two-dimensional planes. Next, we discuss various combinations of the product of quaternions, and give their geometrical and physical interpretations. In Section \ref{Sec:QE} we define a new quaternionic transformation, which causes bipartite entanglement. The proposed map is based on topological and rotational properties of quaternions. We prove that any one-particle quantum state can be transformed into an entangled state by the use of the proposed operation. 

\section{Quaternions and Complexified Quaternions}\label{Sec:Quat}
\subsection{Real Quaternions}
Real quaternions, simply called: quaternions, are generalizations of complex numbers. They were invented by W.R. Hamilton in 1883 to represent rotations in $\mathbb{R}^3$. The algebra of  quaternions--$\mathbb{H}$ is isomorphic to $\mathbb{R}^4$ and it is a non-commutative division algebra, i.e., any non-zero element in $\mathbb{H}$ has its inverse. Any element of $\mathbb{H}$ can be described in the following way:

\begin{equation}
q=a +b \bhat{i} + c \bhat{j} +d \bhat{k},
\end{equation}

\bigskip

\noindent where: $a,b,c,d \in \mathbb{R}$; $\bhat{i}, \bhat{j}, \bhat{k}$ are anti-commuting operators, i.e., $ \bhat{i} \bhat{j} = - \bhat{j} \bhat{i}$, $ \bhat{j} \bhat{k} =   - \bhat{k} \bhat{j}$, $ \bhat{k} \bhat{i} =  - \bhat{i} \bhat{k}$.The anti-community leads to the following identities:  $\bhat{i}^2=\bhat{j}^2=\bhat{k}^2=\bhat{i}\bhat{j}\bhat{k}=-1$.   

\bigskip Let $ p,  q \in \mathbb{H}$ then each quaternion can be written as a sum of a scalar part and a vector part, i.e.: $ q=S(q)+V(q)$: $S(q) \in \mathbb{R}$, $V(q) \in \mathbb{R}^3$.  The operations of addition and multiplication on quaternions $p,q \in \mathbb{H}$ are defined as follows:

\begin{equation}
p+ q=S(p)+S(q) +V(p)+V(q),
\end{equation}
\begin{equation}
pq=S(p)S(q) - \braket{V(p)}{V(q)}+S(p)V(q)+S(q)V(p)+V(p) \times V(q).
\end{equation}

where: $\mid$ is the inner product, and $\times$--the cross product\footnote{$\braket{V(p)}{V(q)} = \dfrac{1}{2}[V(p)V(q)+V(q)V(p)]=\braket{V(q)}{V(p)}$.}. 

\bigskip 
We can also define an operation of conjugation, by analogy to complex numbers:

\begin{equation}
\overline{ q}=S(q) - V(q);
\end{equation}

The norm of a quaternion and its inverse are defined in the following way:

\begin{equation}
N_q={S(q)}^2+ \braket{V(q)}{V(q)}= q  \overline{ q} =  \overline{ q}   q =S(q \overline{q}),
\end{equation}

We have:

\begin{equation}
N_{pq}=N_p N_q, 
\end{equation}
\begin{equation}
S(pq)=S(qp)
\end{equation}

and 
\bigskip
\begin{equation}
q^{-1}=\dfrac{\overline { q}}{N_q},
\end{equation}
\begin{equation}
({ p   q})^{-1} =  q^{-1}   p^{-1}.
\end{equation}

\bigskip
The angle of rotation--$ \theta $ associated with a quaternion $ q$ is the following:

\begin{equation}
\cos \theta = \dfrac{S(q)}{\sqrt{N_q}}, \; 
\end{equation}
\begin{equation}
\sin \theta = \dfrac{\sqrt{\braket{V(q)}{V(q)}}}{\sqrt{N_q}}.
\end{equation}
\bigskip
If $N_q=1$ then $ q$ is called a unit quaternion. Any $ q$ can be written in a polar form, by analogy to complex numbers as follows:

\begin{equation}\label{polarformrealq}
q =\sqrt{N_q}(\cos \theta + \hat{q} \sin
\theta),
\end{equation}

\bigskip 

where $\hat{q} \in \mathbb{R}^3$ is a unit vector, being the axis of rotation, and $\hat{q}^2=-1$. As a result, the plane $ q=s+v \hat{q}, \; s,v \in \mathbb{R}$, for fixed $\hat{q}$, is isomorphic to the complex plane: $w=s+vi$ and within this plane the quaternion multiplication can be reduced to complex multiplication   \cite{Ward1997}.

\bigskip
The inner product on quaternions is defined in the following way:

\begin{equation}
\braket{ p}{ q}=S(p \overline{q})
\end{equation}

As a result, the angle $\lambda$ between two quaternions $ p,  q$ can be figured out as follows:

\begin{equation}
\cos \lambda = \dfrac{S(p \overline{q})}{\sqrt{N_p} \sqrt{N_q}}.
\end{equation}

\bigskip
If $S(p \overline{q})=0$  then $ p \bot  q$; if $V(p \overline{q})=0$  then $ p \parallel  q$. 

If $S(q)=0$ then $ q$ is called a pure quaternion.

\bigskip

\subsection{Complexified Quaternions}
Complexified quaternions--$\mathbb{B}$, are defined as a tensor product of complex numbers and real quaternions, i.e., $\mathbb{B}=\mathbb{C \otimes H}$.  

\bigskip
Any biquaternion can be written as follows:

\begin{equation}
p= q_1+i  q_2=c_0\bhat{1} +c_1i \bhat{i} + c_2 i \bhat{j} +c_3 i \bhat{k}: \;  q_1,  q_2 \in \mathbb{H}; c_0, c_1,c_2, c_3 \in \mathbb{C}.
\end{equation}

or

\begin{equation}
p=\alpha \bhat{1}+ i \underline{\beta}: \; \alpha \in \mathbb{C}, \; \underline{\beta} \in \mathbb{C} ^3.
\end{equation}

\bigskip

This means that we change the basis from $\{1, \bhat{i}, \bhat{j}, \bhat{k}\}$ to $\{ 1, i \bhat{i}, i \bhat{j}, i \bhat{k}  \}$. In this way the biquaternions $i \bhat{i}$, $i \bhat{j}$, $i \bhat{k}$, viewed in $M_2(\mathbb{C}$) representation become the Pauli matrices: $i \bhat{i} =  \sigma_x$, $i \bhat{j} =  \sigma_y$, $i \bhat{k} =  \sigma_z$. 

\bigskip
$\mathbb{B}$ is a non-commutative and non-division algebra over $\mathbb{C}$ with 4 dimensions or, equivalently, it is an algebra over $\mathbb{R}$ with 8 dimensions  \cite{Herrero2017, PW2017, SMLGB1977, Ward1997,Hestenes1971}.

\bigskip
In $\mathbb{B}$ we have three different conjugations: complex $(\ast)$, quaternion $(-)$ and Hermitian $(\dag)$. Under complex conjugation $i \mapsto -i, \bhat{i} \mapsto \bhat{i}$, $\bhat{j} \mapsto \bhat{j}$, $\bhat{k} \mapsto \bhat{k}$;  under quaternionic conjugation  $i \mapsto i, \bhat{i} \mapsto - \bhat{i}$, $\bhat{j} \mapsto - \bhat{j}$, $\bhat{k} \mapsto - \bhat{k}$ and under Hermitian conjugation: $ i \mapsto -i$, $\bhat{i} \mapsto - \bhat{i}$, $\bhat{j} \mapsto - \bhat{j}$, $\bhat{k} \mapsto - \bhat{k}$.

\bigskip
It can be proved that:
$(\overline{ q})^*=\overline{ q^*}$, $(qp)^*=q^*p^*$, $\overline{ qp}=\overline{ p} \overline{ q}$.

\bigskip
If a complexified quaternion commutes with each other then it is a complex number.

\bigskip
Since we have three different conjugations in $\mathbb{B}$, we can define the inner product in different ways:

\begin{equation}\label{innher}
\braket{p}{q}=S(pq^\dag),
\end{equation}

or

\begin{equation}\label{inncong}
\braket{p}{q}=S(p \overline{q} ).
\end{equation}

The norm equals to:
\begin{equation}
N_p=\braket{p}{p}=\dfrac{1}{2}[pp^\dag + p^\ast \overline{ p}];
\end{equation}

however, for (\ref{innher}), which is applied in quantum mechanics, usually, $N_{pq} \ne N_p N_q$:

\begin{equation}
N_{pq}=
\dfrac{1}{2}[pq^\dag + q^\ast \overline{ p}]=2N_qN_p - N_{pq^\ast}.
\end{equation}

\bigskip

If $q=\pm q^\ast$ then $N_{pq}=N_pN_q$ and $p^{-1}=\dfrac{p^\dag}{N_p}$ \cite{Ward1997}.

\bigskip
As real quaternions, any complexified quaternion can be written in a polar form:

\begin{equation}\label{polarformcomplexq}
q =\sqrt{N_q}(\cos z + \hat{q} \sin
z), \; N_{\hat{q}}=1, \; \hat{q}^2 =-1, \; z \in \mathbb{C}.
\end{equation}
\bigskip
If $\overline{q}=q^\ast$ then $z = i \dfrac{\theta}{2}$ and 

\begin{equation}
\cos z =\cosh \dfrac{\theta}{2}=\dfrac{\alpha}{\sqrt{N_q}},
\end{equation}
\begin{equation}
\sin z = i \sinh \dfrac{\theta}{2}=\dfrac{\sqrt{\braket{\underline{\beta}}{\underline{\beta}}}}{\sqrt{N_q}}.
\end{equation}

\bigskip
\bigskip
\subsection{Rotations in terms of Quaternions}

Let $ x, q \in \mathbb{H}$ and $N_q=1$. Since $\mathbb{H}$ is non-commutative, we will consider two maps: $\phi _L(x)= q    x $ and $\phi _R(x)= x    q $;  $\phi_L , \phi_R : \; \mathbb{H} \rightarrow \mathbb{H}$. For $ q=\cos \theta + \hat{q} \sin \theta$  the angle $\omega$ between $ x$ and $ q    x$ is the following:

\begin{equation}
\cos \omega = \dfrac{S(x \overline{(qx)})}{\sqrt{N_x} \sqrt{ N_{qx}}}=\dfrac{S(x)S(\overline{x} \;  \overline{q})}{N_x}=S(\overline{q})=S(q)= \cos \theta ;
\end{equation}

\bigskip therefore, the angle of rotation $\omega$ is the angle of $ q$. It is possible to break up a rotation in $\mathbb{R}^4$ into two  simultaneous rotations in orthogonal planes. We can verify that:

\begin{equation} 
qx= x \cos \theta + (\hat{q} x) \sin \theta, 
\end{equation}
\begin{equation}
q(\hat{q} x)=\hat{q} x \cos \theta -  x \sin \theta.
\end{equation}

\bigskip
For $ x^{'}= \hat{ q} x$ we will have:

\begin{equation}
S(x \overline{\hat{q}x})=S(-x \overline{x} \hat{q})=- x \overline{ x} S(\hat{q})=- x \overline{ x}0=0;
\end{equation}

\bigskip
hence, $ x \bot  x^{'}$. This means that $\phi_L$ is a counterclockwise rotation of elements in the plane containing $ x,   x^{'}$ through the angle $\theta$, i.e., the plane $a  x +b  x^{'}$, $a,b \in \mathbb{R}$ remains invariant under $\phi_L$. In particular, the plane $a 1 + b \hat{q}$ is invariant under left multiplication. 

\bigskip Now, let $\hat{v}, \hat{w}, \hat{q}$ be pure quaternions, which form a right-handed, mutually orthogonal system. Then, for $ x=\hat{v}$, $ x^{'}=\hat{q}, \; \hat{q} \hat{v}= \hat{w}$ we will have:

\begin{equation}
q \hat{v} = \hat{v} \cos \theta +  \hat{q} \hat{v} \sin \theta = \hat{v} \cos \theta + \hat{w} \sin \theta, 
\end{equation}
\begin{equation}
q \hat{w} = \hat{w} \cos \theta +  \hat{q} \hat{w} \sin \theta =- \hat{v} \sin \theta + \hat{w} \cos \theta ;
\end{equation}

\bigskip 

therefore, elements in the plane containing $\hat{v}, \hat{w}$ have been counterclockwise rotated through the angle $\theta$.

\bigskip
Similarly, for $\phi_R$, we will get  a counterclockwise rotation through the angle $\theta$ in the plane containing $1, \hat{q}$  and a clockwise rotation through the angle $\theta$ in the plane containing $\hat{v}, \hat{w}$. 

\bigskip In summary, any rotation in $\mathbb{R}^4$ can be decomposed into two rotations: a rotation of elements in the plane containing a scalar axis and the vector $\hat{q}$ and a rotation of elements in the plane spanned by pure quaternions perpendicular to $\hat{q}$.  

\bigskip
Now, if we choose a special combination of the left and right multiplication, we can define a rotation in $\mathbb{R}^3$. In fact, it was shown that the map $\phi: \; \mathbb{H} \rightarrow \mathbb{H}$ for $q,x \in \mathbb{H}$, $N_q=1$ defined as follows:

\begin{equation}\label{rotinH}
\phi (x)=qxq^{-1}
\end{equation}

\bigskip

leaves the elements in the plane $a 1 + b \hat{q}$ untouched, while the elements in the plane $a \hat{v} + b \hat{w}$ are counterclockwise rotated through the angle $2 \theta$. This results from the fact that multiplication on the right by $ q^{-1}$ causes a clockwise rotation in the plane $a 1 + b \hat{q}$ and a counterclockwise rotation in the plane $a \hat{v} + b \hat{w}$. In conclusion, $\phi$ describes a rotation of a vector part of $ x$ about the vector part of $ q$, i.e., about $\hat{q}$ being the axis of rotation, through the angle $2 \theta$.

\bigskip

In the algebra of biquaternions, rotations depend on the inner product and of combinations of $q$, its conjugates and $w$, e.g., when the inner product is defined as in (\ref{innher}) then the map:

\begin{equation}
\psi(w)=qw \overline{q}: \; q^\ast = q, \; N_q =1
\end{equation} 

\bigskip 

is equivalent to (\ref{rotinH}), i.e., it rotates the vector part of $w$: $V(w)$ about the axis $V_q$ through twice the angle of $q$. 

\bigskip
If we take the following transformation:
\begin{equation}
\phi (x)= q^\dag x q, \; N_q=1 
\end{equation}

\bigskip then $\phi$ describes the Lorentz Transformation of a single space-time event represented by a complexified quaternion $x$.


\bigskip
Analogously to rotations in $\mathbb{H}$, for $x,q \in \mathbb{B}, \; N_q=1$ elements in the plane $(x, \hat{q}x)$ are rotated through the \underline{complex} angle $z$. If $\hat{v}, \hat{w}, \hat{q}$ form a right-handed system in $\mathbb{R}^3$,  $x=\hat{v}, \; x^{'}=\hat{w}$, then the multiplication on the left by a complexified quaternion $q$ rotates elements in the plane containing $(\hat{v}, \hat{w})$ and in the plane containing $(1, \hat{q})$ through the complex angle $z$. Multiplication on the right by $q$ rotates elements in the plane containing $(\hat{v}, \hat{w})$ through the complex angle $-z$ and in the plane containing $(1, \hat{q})$ through the complex angle $z$. As a result the transformation:

\begin{equation}
\mu (x)= \overline{ q} x q, \; N_q=1
\end{equation}

\bigskip 
is a rotation of $V(x)$ through the complex angle $2z$ about $V(q)$ \cite{Ward1997}.

\bigskip
\bigskip

\section{Quaternionic Entanglement}\label{Sec:QE}

In quantum mechanics each quantum particle is in a state being a superposition of all possible states; therefore, it is entangled, too. In the case of one particle, if it splits and become a bipartite system, it becomes an element of $\mathbb{C}^4$; hence, an element of $\mathbb{B}$. As a result, the algebra of  biquaternions is a natural environment to describe bipartite systems. The space adopted in QM is contained in $\mathbb{B}$ since all vectors representing quantum states are normalized. 

In QM the state of a one particle: $q$ is described by the wave function $ \mid \psi  \rangle$ in the following way: 

\begin{equation}
 \mid \psi \rangle = \alpha \mid 0 \rangle  + \mid \beta \rangle: \; \alpha , \beta \in \mathbb{C}; \;  \mid\alpha \mid ^2+ \mid \beta \mid ^2=1.
\end{equation}

\bigskip 
Now, we embed a state $q$ in $\mathbb{B}$ to have a representation of a bipartite state; therefore, we will exam four possible representations of $q$:

\begin{equation}\label{q1}
q=(\alpha, \beta,0,0),
\end{equation}

\begin{equation}\label{q2}
q=(0,0,\alpha, \beta),
\end{equation}

\begin{equation}\label{q3}
q=(\alpha, 0,\beta,0),
\end{equation}

\begin{equation}\label{q4}
q=(0,\alpha,0, \beta).
\end{equation}
\bigskip

We exclude the states $(0,\alpha,\beta,0), \; (\alpha,0,0,\beta)$ since they are already entangled. To measure entanglement, we will apply the concept of concurrency--C \cite{Wootters2001,Fano1957}\footnote{$q=(q_1,q_2, q_3, q_4)$ is entangled iff $C=2 \mid q_1 q_4 - q_2q_3 \mid \ne 0$. If $C=1$ then $q$ is maximally entangled.}.We also do not consider pure states in $\mathbb{B}$ because they can be continuously transformed into mixed states in the following way:

\bigskip 

$(\alpha,0,0,0)(0,0,a_3,a_4)= (0,0, \alpha a_3, \alpha a_4)$,

$(0,\alpha,0,0)(0,0,a_3,a_4)= (0,0,-\alpha a_4,\alpha a_3)$, 

$(0,0,\beta,0)(0,0,a_3,a_4)= (-\beta a_3,\beta a_4,0,0)$, 

$(0,0,0,\beta)(0,0,a_3,a_4)= (- \beta a_4,- \beta a_3,0,0)$.

\bigskip 

The above assumption is in the case when we deal with pure states in $\mathbb{C}^2$.

\bigskip  Now, let $p$ be a quaternion of rotation, necessary for causing entanglement. We propose the map $\Lambda : \; \mathbb{B} \rightarrow \mathbb{B}$ as an operator of entanglement defined as follows:

\begin{equation}
\Lambda (q) = pqp,
\end{equation} 

\bigskip 
\noindent where $N_p=1, \; p=(a_1,a_2,a_3,a_4), \; a_i \in \mathbb{R}, \; i=1,...,4$. 

\bigskip The combination of $pq$ is a counterclockwise rotation of $q$ through the angle $\theta$ in the plane spanned by a scalar and a vector $1, \hat{ q}$, and in the plane spanned by orthogonal vectors $\hat{v}, \hat{w}$. The  product $qp$ is a counterclockwise rotation of $q$ through the angle $\theta$ in the plane spanned by $1, \hat{q}$, and a clockwise rotation through the angle $\theta$ in the plane $\hat{v}, \hat{w}$. The proposed combination of rotations was chosen in order to break one common dimension between $p$ and $q$. We cannot indicate  within  which plane it happens because it depends on the combination of the analyzed  quaternions. Finally, in the above definition of $\Lambda$ the symmetry property has been taken into account. In fact, when an experiment is made, we only act on a particle which is to be split.

\bigskip 
We mentioned above that we want to break one common dimension between $p$ and $q$; hence, we have to put some restrictions on  $p$ and $q$:
\begin{enumerate}
\item [(R1)]$p$ is not entangled. This restriction seems obvious since we want to create entanglement.
\item [(R2)] $p$ is not a pure state. In such a case the map $\:\Lambda (q)$ would act either as a reflection in four dimensional space or as identity, e.g. $q=(\alpha,\beta,0,0), \; p=(0,1,0,0)$ $\Lambda(q)=-q$; $q=(\alpha,\beta,0,0), \; p=(1,0,0,0)$ $\Lambda(q)=q$.
\item [(R3)] $p$ is connected to $q$ only in one direction. This means that pairs of orthogonal planes for $q$ and $p$ intersect; therefore, if we perform an operation $\Lambda$, we hope to break the direction being in common. 
\end{enumerate}

\begin{example}
$p=(1/\sqrt{2},0,1/\sqrt{2},0)$, $q=(i/\sqrt{2},-i/\sqrt{2},0,0)$,

$\hat{q}=(1,0,0)$, $\hat{v}=(0,1,0)$, $\hat{w}=\hat{q} \hat{v}=(0,0,1)$, 

$\Lambda(q)=(0,-i/\sqrt{2},i/\sqrt{2},0)$--the scalar direction has been broken.
\end{example}

\begin{example}
$p=(0,0,1/\sqrt{2},1/\sqrt{2})$, $q=(i/\sqrt{2},0,-i/\sqrt{2},0)$,

$\hat{q}=(0,1,0)$, $\hat{v}=(0,1/\sqrt{2},1/\sqrt{2})$, $\hat{w}=\hat{q} \hat{v}=(1,0,0)$, 

$\Lambda(q)=(-i/\sqrt{2},0,0,-i/\sqrt{2})$--the $\hat{q}$ direction has been broken.
\end{example}

\begin{example}
$p=(0,0,1/\sqrt{2},1/\sqrt{2})$, $q=(0,i/\sqrt{2},0,i/\sqrt{2})$,

$\hat{q}=(1/\sqrt{2},0,1/\sqrt{2})$, $\hat{v}=(0,1/\sqrt{2},1/\sqrt{2})$, $\hat{w}=\hat{q} \hat{v}=(-1/\sqrt{3},-1/\sqrt{3},1/\sqrt{3})$, 

$\Lambda(q)=(0,i/\sqrt{2},-i/\sqrt{2},0)$

the common direction in the forth coordinate has been broken. 

This direction is determined by both planes, i.e., $(1,\hat{q})$ and $(\hat{v},\hat{w})$.
\end{example}

Two dimensions in common would violate (R1).

\bigskip 

Assuming that the coefficients of $p$ taken into consideration are not equal to zero, the quaternions: 

$(a_1,a_2,a_3,0),(a_1,a_2,0,a_4)$, 

$(a_1,0, a_3,a_4), (0,a_2,a_3,a_4)$,

$(a_1,0,0,a_4), (0,a_2,a_3,0)$ 

are entangled; therefore, they do not fulfill (R1); 

the quaternion: $(a_1,a_2,a_3,a_4)$ does not fulfill (R3). 

\bigskip As a result, the following statement follows:

\bigskip 

\begin{theorem}The map $\Lambda$
operates bipartite entanglement in $\mathbb{B}$  under (R1)-(R3).
\end{theorem}
\begin{proof} The proof consists of simple calculations. We have two possible choices of $p$ for each $q$:
\begin{enumerate}
\item   $q=(\alpha, \beta,0,0)$, $p=(a_1,0,a_3,0)$,

$pqp=(\alpha(a_1^2-a_3^2), \beta, 2 \alpha a_1 a_3, 0) $, 
$C=4\mid \alpha \beta a_1 a_3 \mid \ne 0$.

\item   $q=(\alpha, \beta,0,0)$, $p=(0,a_2,0,a_4,0)$,

$pqp=(-\alpha, -\beta (a_2^2-a_4^2),0, -2 \beta a_2 a_2, 0) $,	 		
$C=4\mid  \alpha \beta a_2 a_4 \mid \ne 0$.

\item   $q=(0;0;\alpha, \beta)$,  $p=(a_1,0,a_3,0)$,

$pqp=(-2\alpha a_1 a_3;0; \alpha (a_1^2-a_3^2); \beta) $,	
$C=4 \mid  \alpha \beta a_1 a_3 \mid \ne 0$.

\item    $q=(0;0;\alpha, \beta)$, $p=(0,a_2,0,a_4)$,

$pqp=(0,-2\beta a_2 a_4, \alpha,- \beta (a_4^2-a_2^2)) $,		 
$C=4 \mid \alpha \beta a_2 a_4 \mid \ne 0$. 

\item    $q=(\alpha,0, \beta,0)$, $p=(a_1,a_2,0,0)$,

$pqp=(\alpha a_1^2-\alpha a_2^2,2\alpha a_1 a_2,  \beta,0) $,		 
$C=4 \mid \alpha \beta a_1 a_2 \mid \ne 0$. 

\item    $q=(\alpha,0, \beta,0)$, $p=(0,0,a_3,a_4)$,

$pqp=(-\alpha,0,-\beta a_3^2 + \beta a_4^2,2\beta a_3 a_4) $,		 
$C=4 \mid \alpha \beta a_3 a_4 \mid \ne 0$. 

\item    $q=(0,\alpha,0, \beta)$, $p=(a_1,a_2,0,0)$,

$pqp=(-2\alpha a_1 a_2, \alpha a_1^2-\alpha a_2^2,0,  \beta) $,		 
$C=4 \mid \alpha \beta a_1 a_2 \mid \ne 0$. 

\item   $q=(0,\alpha,0, \beta)$,  $p=(0,0,a_3,a_4)$,

$pqp=(0,\alpha,-2\beta a_3 a_4,\beta a_3^2 - \beta a_4^2) $,		 
$C=4 \mid \alpha \beta a_3 a_4 \mid \ne 0$. 		
\end{enumerate}

\end{proof}

\bigskip

We can observe that we always obtain the same degree of entanglement equal to $4 \mid \alpha \beta a_i a_j \mid$. If $\mid \alpha \beta a_i a_j \mid=1/4$ then we get maximally entangled states. For example, this happens for  $\alpha=\beta=a_1=a_2=a_3=a_4=1/\sqrt{2}$. Such parameters determine quaternions representing rotations through the angle $\pi /4$. It is not surprising here that we obtained exactly the same angle of rotation which was applied between the calcium bean and lens in the experiment of Kocher \cite{Ko1967}.

\bigskip

\section{Discussion}
In the proposed work we defined an entanglement map--$\Lambda$ which, under several restrictions, transforms a quantum one-particle state into a bipartite entanglement. In our definition we relied on rotational properties of quaternions, and on the fact that quantum superposition means that a quantum state is in all possible states at the same time. We are aware that the presented method is totally different than methods applied in QM. Such descriptions use local quantum operations and classical communications--LOCC. In fact, in classical method, e.g., the requirement for bipartite entanglement is that the Schmidt rank of a state is greater than or equal to the Schmidt rank of the unitary operator $U$ \cite{DSRG2011}. In the model proposed in this manuscript, always, the Schmidt rank of a state $\mid \psi \rangle$ is equal to the Schmidt rank of the unitary operator $p$; hence we are in full agreement with other approaches. 

\bigskip
The proposed quaternionic description of bipartite entanglement is important because in all physical experiments which were oriented to make more particles entangled, always the point of departure was a pair of two entangled particles \cite{GHZ1990,Zet1997}. Therefore, the quaternionic mechanism of bipartite entanglement can shed new light for various experiments. We think, that it is not necessary to invent an octonic description for three entangled particles, but it is enough to know the most basic quaternionic entangled relationship and apply the knowledge to the tensor product for more than two particles. Such studies are in progress. We  hope that this proposal can open fascinating perspectives not only for quantum computing, but for the whole quantum mechanics. 

\bigskip 

\section*{Acknowledgments}

This work was performed thanks to the financial support of the Polish Ministry of Arts and Higher Education, no. 493/S/17.

\bigskip
\bigskip
\bibliographystyle{plainnat}

\bigskip
\bigskip

\end{document}